\documentclass[oneside,english]{amsart}
\usepackage[T1]{fontenc}
\usepackage[latin9]{inputenc}
\usepackage{amsthm}
\usepackage{amssymb}

\makeatletter
\numberwithin{equation}{section}
\numberwithin{figure}{section}
\theoremstyle{plain}
\newtheorem{thm}{Theorem}

\makeatother

\usepackage{babel}

\begin{document}

\title{Gods as Topological Invariants}

\author{Daniel Schoch}

\date{April 1st, 2012}
\begin{abstract}
We show that the number of gods in a universe must equal the Euler
characteristics of its underlying manifold. By incorporating the classical
cosmological argument for creation, this result builds a bridge between
theology and physics and makes theism a testable hypothesis. Theological
implications are profound since the theorem gives us new insights
in the topological structure of heavens and hells. Recent astronomical
observations can not reject theism, but data are slightly in favour
of atheism.
\end{abstract}

\keywords{Topology, Euler characteristics, manifolds, invariants, mathematical
theology, mathematical joke.}

\maketitle

\section{Motivation}

Conventional theology builds on faith and metaphysical assumptions.
While faith is unquestionable, metaphysics is generally considered
untestable. This lead to the widespread assumption that theology and
natural science form non-overlapping magistrates. However, nothing
could be further away from the religious tradition. Both Arabic and
Christian medival thinker were intending a synthesis between the Aristotelic
scientific view and their specific religion, in particular the Abrahamitic
monotheism. Up to the early modern times, Galileis processes and the
rejection of atomism by the Catholic church have revealed possible
conflict zones between both fields. 

Although metaphysical speculations always linked theological and mathematical
concepts, the introduction of distinctive mathematical methods to
support theism has long been attributed to Euler. In his famous alledged
encounter with Diderot, Euler presented a mock algebraic proof to
embarrass the atheist philosopher Diderot, which is probably a 19th
century legend. The true core of the legend, however, might be that
the possibility of an algebraic proof of the existence of god was
discussed among 18th century intellectuals \cite[p. 129]{STR67}.

The first successful step in mathematization of theology was taken
by Goedel \cite{SOB87}. His formalization of the ontological argument,
originally formulated in terms of modal logic, has been reconstructed
in set-theoretical form \cite{ESS93}. It was noticed that the set
structure was identical to the topological concept of an ultrafilter,
which led to speculations about the specific relations between theology
and topology \cite{CAL95}. 

It is the aim of this paper to advance these a priori speculations
and convert them to a testable theory, linking theology and cosmology
very much in the spirit of the medival Christain tradition. Although
current data is not deceisive in this matter and is slightly supporting
a zero-god universe, the path has been paved for a fruitful interdisciplinary
collaboration of physics, mathematics, philosophy and theology.

\section{The Cosmological Argument}

The cosmological argument (Platon, Philoponos, Aquinas etc.) states
that a First Cause of the universe must exist, since no causal chains
are finite and do not contain loops. Insofar it is based on the simple
fact that any component of an ordered and circle-free finite graph
is a tree and thus must have a first vertex. To apply this to a pre-Einstein
universe one must add the assumptions that matter does not by itself
emerge out of the vacuum, that no motion can occur out of rest, and
that initial or primordial matter is at rest. It follows that a First
Mover must have initiated all cosmic motion at the moment the universe
comes into existence. Since it is not by itself a natural phenomenon,
such a mighty cause can only be a god.

However, this argument does not remain valid if the universe has origined
from a Big Bang. If the universe expands out of a zero-volume point
at $t=0$, there no causal relation can connect an event at some time
$t>0$ of the existing universe by an event $t<0$ before the Big
Bang. In the watchmaker analogy, the watch could not have been wound
up.

We nevertheless embrace the result of the argument, as far as it is
restricted to the types of steady-state universes for which it was
originally intended. It goes without saying that anything beyond the
three-dimensional Euclidean space was out of imagination for the medieval
scholar. Before Einstein, time was considered absolute and independent
of space and matter. A physical explanation for a universe emerging
out of nothing was unthinkable and incompatible with the mechanics
of their time, may it be Aristotelian, Galileian or Newtonian. The
initial singualarity of an Einstein-Friedman universe is, however,
a distinctive topological feature of the manifold itself. We assume
therefore, in accordance with the cosmological argument, that a finite
Aristotelian universe, which manifold can be desribed by a compact
subset of $\mathbb{R}^{3}$ homeomorphic to a ball (a 3-cell), has
one and only one god.

\section{The Main Theorem}

Let $U$ be a unviverse with underlying manifold $M_{U}$. By $\Theta\left(U\right)$
we denote the number of gods in the universe. We postulate the following
axioms. 

From the cosmological argument we obtain

\textbf{Axiom 1.} The number of gods in a 3-cell is one.

Gods are eternal, invariable and do not depend on the evolution of
the cosmos under the laws of physics. Changing the latter would have
had no influence on the existence of gods. Since the laws of physics
are continuous, if $M_{0}$ and $M_{1}$ are the underlying manifolds
of the spatial universe at some time $t_{0}$ and $t_{1}$ (in comoving
coordinates), respectively, they must be homotopy equivalent. We therefore
obtain

\textbf{Axiom 2.} The number of gods is a homotopy invariance.

Clearly, each god only belongs to only one universe. In a multiverse
consisting of disjoint unions of universes, the number of gods must
therefore satisfy additivity.

\textbf{Axiom 3.} Let $U$ and $U'$ be separated universes with compact
manifolds. Then the number of gods in the disjoint union is the sum
of the numbers of gods of the parts,\begin{equation}
\Theta\left(U\sqcup U'\right)=\Theta\left(U\right)+\Theta\left(U'\right).\label{eq:Sum}\end{equation}

\begin{thm}
The number of gods in a universe equals the Euler characteristics
of the underlying manifold,\[
\Theta\left(U\right)=\chi\left(M_{U}\right).\]
\end{thm}
\begin{proof}
By the second axiom, the number of gods only depend on the underlying
manifold. Thus we can write w.l.o.g. $\Theta\left(U\right)=\Theta\left(M_{U}\right)$.
Since the $n-$cell is null homotope, by axiom 1 and 2 we obtain\begin{equation}
\Theta\left(\mbox{pt}\right)=1.\label{eq:Pt}\end{equation}
Let further $M$ and $N$ be any compact sets. We construct a homotopy
transforming $M\cup N$ into separated copies $A,B,C$ of the closure
of $M\setminus N$, $M\cap N$, and $M\setminus N$, respectively.
Equation (\ref{eq:Sum}) together with axiom 2 implies\begin{eqnarray*}
\Theta\left(M\cup N\right) & = & \Theta\left(A\right)+\Theta\left(B\right)+\Theta\left(C\right),\\
\Theta\left(M\right) & = & \Theta\left(A\right)+\Theta\left(B\right),\\
\Theta\left(N\right) & = & \Theta\left(B\right)+\Theta\left(C\right),\\
\Theta\left(M\cap N\right) & = & \Theta\left(B\right).\end{eqnarray*}
We obtain the inclusion-exclusion principle,\[
\Theta\left(M\cup N\right)=\Theta\left(M\right)+\Theta\left(N\right)-\Theta\left(M\cap N\right).\]
 By a well-known characterization, the latter together with (\ref{eq:Pt})
implies $\Theta\left(M\right)=\chi\left(M\right)$ for all compact
manifolds $M$.
\end{proof}
The product theorem for the Euler characteristics implies that for
manifolds $M$ and $N$\[
\Theta\left(M\otimes N\right)=\Theta\left(M\right)\cdot\chi\left(N\right).\]
In particular, if time is itself an interval $T\subset\mathbb{R}$,
we find that the number of gods are independent of time,\[
\Theta\left(M\otimes T\right)=\Theta\left(M\right).\]
This is compatible with the scholastic view introduced by Boetius
in his \emph{Consolations} that god is above time. However, the same
formula spells trouble for all theologies which are based on a cyclic
conception of time, which is widespread in India (Veda) and among
native American religions. Since $\chi\left(S^{1}\right)=0$ there
are no gods in such a universe,\[
\Theta\left(M\otimes S^{1}\right)=0.\]

\section{Universes, Heavens and Hells}

The number of gods has come out to be an integral number. This rules
out any demi-gods or lower devas, as they are known from Greek or
Indian mythology. However, the divine cardinal can get negative; with
the obvious interpretation of these gods being devils. By the additivity
theorem, components of positive and negative Euler characteristics
could cancel each other out. We can safely assume, however, (and there
is plenty of support from religious texts) that gods and devils can
not have stable coexistence in the same part of the universe. Thus
each component contains only gods or devils, but not both. The absolute
value of the Euler characteristics of the universe therefore equals
the number of supreme and most inferior beings in it, dependent on
the sign.

A lot of types of universes are godless. These are all spaces which
contain the 1-sphere (circle) as a factor, such as the tori and all
products of them with an arbitrary manifold. This also applies if
the universe is infinite Euclidean, but has additional warped dimensions,
as suggested in string theory. Also, a spherical three-dimensional
universe would have no gods. The only non-exotic topology with a positive
number of gods are Euclidean spaces, which all contain exactly one
god, which is well in accordance with the Jewish-Christian and Islamic
tradition.

An interesting theoretical question concerns the topological structure
of heaven. The $3$-dimensional Eucledian space is suitable, but since
souls are immaterial, they are not confined to three-dimensionality.
It is unlikely, however, that heaven is a bounded manifold, since
there should be no limit in heaven. One possible structure of a monotheistic
heaven is the real projective plane. Another is the 2-sphere, but
it requires a pair of gods. The 2-spehere would have been a preferred
choice of Greek philosophers, since its imbedding in the $\mathbb{R}^{3}$
is in perfect alignment with the Pythagorean-Platonic idea of a perfect
body, which was influential in the early Scholastics. A suitable pair
of gods entrenched in the Abrahamitic tradition is Jahwe and Asherah,
but apparently the couple broke up some time after the late Bronze
age \cite{FIN01}.

Constraints by traditional religion are more relaxed insofar there
is no dogma imposing an upper boundary for the number of devils. As
in the case of heavens, souls are not restricted to exist only in
spaces with dimesion at least three. If hells are two-dimensional
closed surfaces with topological genus $g>1$, then the Euler characteristics
is $2-2g$, which would correspond to a hell with an even number of
$2\left(g-1\right)$ devils. Such hells can best be envisioned as
multiple tori. The double torus in the form of the figure 8 has Euler
characteristics -2. Each torus attached decrease the Euler characteristics
by a further -2. This suggest that the rings of hell are not concentric,
as Dante speculated, but that they are lined up such that the soul
transgresses through a complete half ring of hell before entering
the next level. This is a more hellish scenario than Dante's concentric
model and thus more realistic.

\section{Evidences}

The topology of the universe could in principle be observed if it
is finite and small (respectively old) enough to have light travelled
through it. In this case one would observe multiple images of the
same constellations of a matter, which for each point source of light
takes the form of a circle. However, this method is not discriminative
for a large or young universe, where it only yields a lower bound
for the size of the univese \cite{BIE11}. A test for infinity of
the universe in one or several dimensions can be based on statistical
analysis of the temperature fluctiations of the background radiation,
which is a remainder of the Big Bang. If one or more dimensions of
the space are topologica circles, space remains homogenuous, but isotropy
is violated.

Unfortunately, topology is just constrained, but not determined by
local curvature. Data from the Wilkinson microwave anisotropy probe
suggest that the universe is flat with only 0.5\% margin of error
\cite{COR04}. A flat universe can have vanishing total energy consistent
with an origin from nothing. An infinite Euclidean space fits the
data. Some exotic topologies such as the Poincaré homology sphere
and the Picard horn have been claimed consistent with the findings,
but for the former this has been challenged \cite{KEY07}. A recent
statistical analysis on the number of infinite dimensions compared
the Euclidean space $\mathbb{R}^{3}$, the 3-torus $T^{3}$ and the
manifolds $T^{2}\otimes\mathbb{R}$ and $S^{1}\otimes\mathbb{R}^{2}$
\cite{ASL11}. Only the Euclidean space has a non-vanishing Euler
characteristics. The most probable topology of the universe was found
to be $T^{2}\otimes\mathbb{R}$, which would support the atheist view
brought forward by many leading cosmologists.

\end{document}